\documentclass[a4paper,UKenglish]{lipics-v2018}

\usepackage{microtype}

\usepackage{graphicx}
\usepackage{pifont}
\usepackage{verbatim}
\usepackage{amsmath,amsfonts}
\usepackage{enumerate,url}
\usepackage{paralist}
\usepackage{wrapfig}
\usepackage{multirow}
\usepackage{todonotes}
\usepackage{array,wrapfig}
\usepackage{algorithm}
\usepackage{algorithmicx}
\usepackage[noend]{algpseudocode}

\newcommand{\ptas}{$\mathsf{PTAS}$ }
\newcommand{\wit}{\mathsf{w}}
\newcommand{\pat}{\mathsf{L}}
\newcommand{\cor}{\texttt{cor}}
\newcommand{\dis}{\texttt{dist}}
\newcommand{\optX}{\mathsf{OPT_X}}
\newcommand{\optY}{\mathsf{OPT_Y}}
\newcommand{\opt}{\mathsf{OPT}}
\newcommand{\apx}{$\mathsf{APX}$}
\newcommand{\elf}{$\mathsf{L}$}
\newcommand{\np}{$\mathsf{NP}$}
\newcommand{\p}{$\mathsf{P}$}
\newcommand{\mds}{$\mathsf{MDS}$}

\newtheorem{obs}[theorem]{Observation}

\title{Approximating Dominating Set on Intersection Graphs of Rectangles and L-frames}
\titlerunning{Approximating Dominating Set on Rectangles and L-frames} 

\author{Sayan Bandyapadhyay}{Department of Computer Science, University of Iowa, Iowa City, USA.}{sayan-bandyapadhyay@uiowa.edu}{}{The work was partially done when the author was visiting University of California, Santa Barbara}
\author{Anil Maheshwari}{School of Computer Science, Carleton University, Ottawa, Canada.}{anil@scs.carleton.ca}{}{}
\author{Saeed Mehrabi}{School of Computer Science, Carleton University, Ottawa, Canada.}{saeed.mehrabi@carleton.ca}{}{}
\author{Subhash Suri}{Department of Computer Science, UC Santa Barbara, California, USA.}{suri@cs.ucsb.edu}{}{}

\authorrunning{S. Bandyapadhyay, A. Maheshwari, S. Mehrabi, and S. Suri} 

\Copyright{Sayan Bandyapadhyay, Anil Maheshwari, Saeed Mehrabi, and Subhash Suri}


\relatedversion{}

\funding{Research of Sayan Bandyapadhyay and Subhash Suri was supported in part by the NSF grant CCF-1525817. Research of Anil Maheshwari is supported in part by NSERC. Saeed Mehrabi is supported by a Carleton-Fields postdoctoral fellowship.}

\nolinenumbers 

\begin{document}

\maketitle

\begin{abstract}
We consider the \emph{Minimum Dominating Set} (\mds) problem on the intersection graphs of geometric objects. Even for simple and widely-used geometric objects such as rectangles, no sub-logarithmic approximation is known for the problem and (perhaps surprisingly) the problem is \np-hard even when all the rectangles are ``anchored'' at a diagonal line with slope -1 (Pandit, CCCG 2017). In this paper, we first show that for any $\epsilon>0$, there exists a $(2+\epsilon)$-approximation algorithm for the \mds~problem on ``diagonal-anchored'' rectangles, providing the first $O(1)$-approximation for the problem on a non-trivial subclass of rectangles. It is not hard to see that the \mds~problem on ``diagonal-anchored'' rectangles is the same as the \mds~problem on ``diagonal-anchored'' \emph{\elf-frames}: the union of a vertical and a horizontal line segment that share an endpoint. As such, we also obtain a $(2+\epsilon)$-approximation for the problem with ``diagonal-anchored'' {\elf-frames}. 
%
On the other hand, we show that the problem is \apx-hard in case the input {\elf-frames} intersect the diagonal, 
or the horizontal segments of the \elf-frames intersect a vertical line. 
However, as we show, the problem is linear-time solvable in case the \elf-frames intersect a vertical as well as a horizontal line. Finally, 
we consider the \mds~problem in the so-called ``edge intersection model'' and obtain a number of results, answering two questions posed by Mehrabi (WAOA 2017).
\end{abstract}

\section{Introduction}
\label{sec:introduction}
\emph{Minimum Dominating Set} (\mds) is an \np-hard problem in graph theory and discrete optimization. Given a graph $G=(V,E)$, the objective of the \mds~problem is to compute a minimum-size subset $V'\subseteq V$ such that every vertex not in $V'$ is adjacent to at least one vertex in $V'$. For general graphs, it is known that a greedy algorithm for \mds~achieves an $O(\log |V|)$-factor approximation and within a constant factor this is the best one can hope for, unless \p=\np~\cite{RazS97}. As such, the problem has been extensively studied on many subclasses of graphs, one of which is the intersection graphs of geometric objects in the plane~\cite{DeL16,Marx06,GibsonP10,ErlebachL08,Mehrabi17,Pandit17}. Here, each vertex of the graph is in one-to-one correspondence with a geometric object in the plane and two vertices are adjacent if and only if the corresponding objects have a non-empty intersection.

In this paper, we consider the approximability and hardness of the \mds~problem on the intersection graphs of geometric objects. The \mds~problem is known to admit $\mathsf{PTAS}$es on disk graphs~\cite{GibsonP10} and the intersection graphs of non-piercing\footnote{Two connected objects $A$ and $B$ are called non-piercing if both $A\setminus B$ and $B\setminus A$ are connected.} objects~\cite{GovindarajanRRR16}. On the other hand, it is \np-hard to obtain a $o(\log |V|)$-approximation in polynomial time for sufficiently complicated shapes, e.g. rectilinear polygons~\cite{DinurS14,ErlebachL08}.
However, even for simple shapes such as axis-parallel rectangles no sub-logarithmic approximation is known. The only approximation for rectangles we are aware of is due to Erlebach et al.~\cite{ErlebachL08} who gave an $O(c^3)$-approximation on rectangles with aspect-ratio at most $c$. In fact, the problem is \apx-hard~\cite{ErlebachL08} on rectangles, and (perhaps surprisingly) the problem is shown to be \np-hard even on diagonal-anchored rectangles~\cite{Pandit17}; that is, the intersection of every rectangle and a diagonal line with slope -1 is exactly one corner of the rectangle. See Figure~\ref{fig:graphs}(a) for an example. However, to the best of our knowledge no sub-logarithmic approximation is known even in this case. We note that optimization problems on ``diagonal-intersecting'' geometric objects have been studied before through the lenses of approximation algorithms; e.g. maximum independent set~\cite{CorreaFPS15,CatanzaroCFHHHS17,KeilMPV17} and minimum hitting set~\cite{CorreaFPS15,ChepoiF13,MudgalP15}.

\subparagraph{Our results.} In this paper, we first give a $(2+\epsilon)$-approximation algorithm for the \mds~problem on diagonal-anchored rectangles, providing the first $O(1)$-approximation for the problem on a non-trivial subclass of rectangles.
\begin{theorem}
\label{thm:twoPlusEpsilon}
For any $\epsilon>0$, there exists a $(2+\epsilon)$-approximation algorithm for the \mds~problem on diagonal-anchored rectangles.
\end{theorem}

To prove Theorem~\ref{thm:twoPlusEpsilon}, we first divide the problem into two subproblems and then give a \ptas for the subproblems using the \emph{local search technique}~\cite{ChanH12,MustafaR10}. Each such subproblem involves diagonal-anchored rectangles that lie on only one side of the diagonal. The key to obtain our \ptas is in showing a planar drawing of a bipartite graph that is required for the analysis of the local search algorithm. We note that, even in these simpler cases the problem remains sufficiently challenging due to the geometry of the rectangles, and the existing schemes are not useful to obtain a near-optimal solution. For example, the local search analyses for non-piercing objects in \cite{GovindarajanRRR16} do not hold here, as the diagonal-anchored rectangles can still ``pierce'' each other. 

\begin{figure}[!h]
\centering
\includegraphics[width=0.50\textwidth]{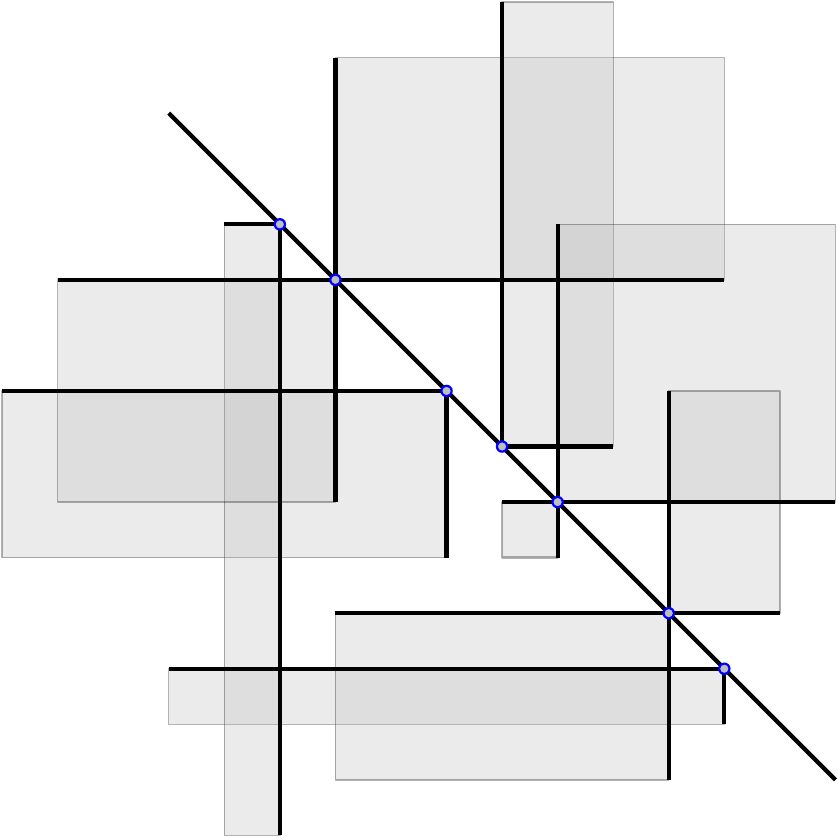}
\caption{The class of diagonal-anchored rectangles is equivalent to that of diagonal-anchored \elf-frames.}
\label{fig:equivalent}
\end{figure}

It is not hard to see that the \mds~problem on ``diagonal-anchored'' rectangles is the same as the \mds~problem on ``diagonal-anchored'' \emph{\elf-frames}~\cite{CatanzaroCFHHHS17}. An \elf-frame is the union of a vertical and a horizontal line segment that share an endpoint (\emph{corner}). Indeed, each rectangle in the input can be replaced by a diagonal-anchored {\elf-frame} without altering the underlying intersection graph (see Figure~\ref{fig:equivalent} for an illustration). Thus, a dominating set for the instance with the {\elf-frames} is also a dominating set for the original instance with rectangles, and vice versa. Hence, we also obtain a $(2+\epsilon)$-approximation for the problem with diagonal-anchored {\elf-frames}.
For the \mds~problem on general \elf-frames, the only approximation we are aware of is due to Mehrabi~\cite{Mehrabi17} who gave an $O(1)$-approximation algorithm when every two \elf-frames intersect in at most one point. Asinowski et al.~\cite{AsinowskiCGLLS12} proved that every circle graph is an intersection graph of \elf-frames. Since \mds~is \apx-hard on circle graphs~\cite{DamianP06}, the problem is also \apx-hard on \elf-frames.

Note that, by definition, we can have four different \emph{types} of \elf-frames depending on the two endpoints that define the corner. Considering the problem on \elf-frames we extend the \apx-hardness result in the general case to two constrained cases. First, we show that the problem does not admit a $(1+\epsilon)$-approximation for any $\epsilon > 0$ on ``diagonal-intersecting'' \elf-frames (see Figure~\ref{fig:graphs}(b) for an example).
\begin{theorem}
\label{thm:apxDiagonal}
The \mds~problem is \apx-hard on \elf-frames when every \elf-frame intersects a diagonal line. 
\end{theorem}

As the construction in proving Theorem~\ref{thm:apxDiagonal} shows, the theorem holds even if the input consists of only one type of \elf-frames and all intersection points of the \elf-frames lie on only one side of the diagonal. This is in contrast to the diagonal-anchored {\elf-frames} case where we obtain a \ptas. We also show that one cannot hope for a $(1+\epsilon)$-approximation for any $\epsilon > 0$ even when all the \elf-frames intersect a vertical line; see Figure~\ref{fig:graphs}(c) for an example. We refer to these \elf-frames as \emph{vertical-intersecting} \elf-frames.
\begin{theorem}
\label{thm:apxVertical}
The \mds~problem is \apx-hard on {vertical-intersecting} \elf-frames even if all the \elf-frames intersect the vertical line from one side. Moreover, the problem is \np-hard even if for each \elf-frame, the horizontal and vertical segments have the same length.
\end{theorem}

Moreover, we show that the \apx-hardness of Theorem~\ref{thm:apxVertical} is almost tight in the sense that the problem admits a polynomial-time algorithm on \elf-frames, where each \elf-frame intersects a vertical line and a horizontal line. See Figure~\ref{fig:graphs}(d) for an example. Note that, all the \elf-frames in the input are of the same type. 
\begin{theorem}
\label{thm:xyCrossingAlgorithm}
The \mds~problem is linear-time solvable on \elf-frames that intersect a vertical line and a horizontal line.
\end{theorem}

To prove Theorem~\ref{thm:xyCrossingAlgorithm}, we show that this class of graphs are the same as permutation graphs for which given the permutation of the vertices, the \mds~problem can be solved in linear time~\cite{ChaoHL00}. As given a set of \elf-frames that intersect a vertical line and a horizontal line, the corresponding permutation can be computed in linear time, the theorem follows.

\begin{figure}[t]
\centering
\includegraphics[width=1.00\textwidth]{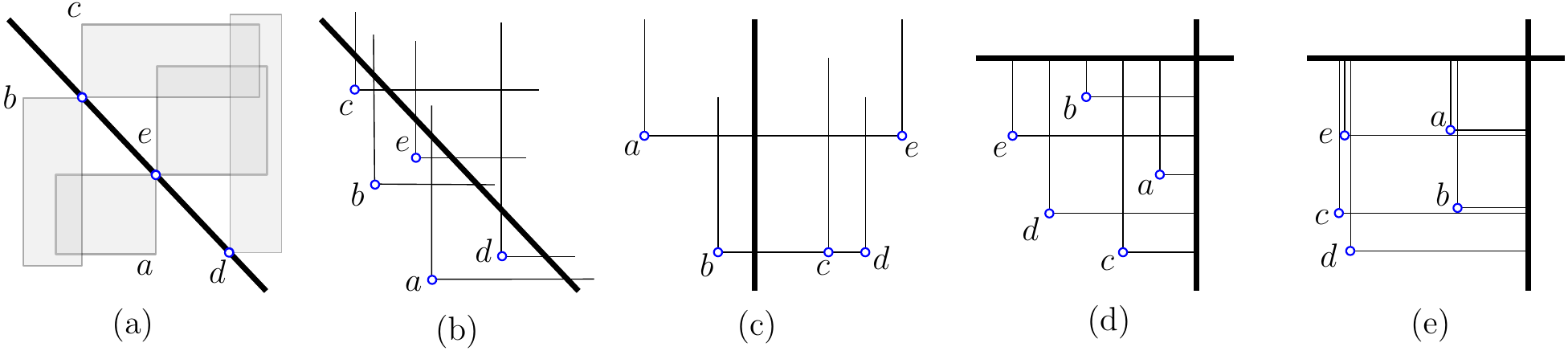}
\caption{A graph $G=(\{a,b,c,d,e\}, E)$ with five different representations, where $E=\{(a,b), (a,e), (b,c), (c,d), (c,e), (d,e)\}$.}
\label{fig:graphs}
\end{figure}

While the standard notion of intersection dates back to 1970s, Golumbic et al.~\cite{GolumbicLS09} introduced the notion of \emph{edge intersection} of \elf-frames. In this model, two \elf-frames are considered adjacent if and only if they overlap in strictly more than a single point in the plane. More formally, the \elf-frames corresponding to the vertices of the graph are drawn on a grid and two vertices are adjacent in the graph if and only if their corresponding \elf-frames share at least one grid edge; see Figure~\ref{fig:graphs}(e) for an example. To distinguish between the two models we will explicitly refer the edge intersection model whenever discussing a result on this model. Otherwise, we always mean the standard intersection model.

For the edge intersection model, there is a 4-approximation algorithm for \mds~on \elf-frames~\cite{ButmanHLR10,HeldtKU14}. Moreover, the problem was recently shown to be \apx-hard by Mehrabi~\cite{Mehrabi17}, where two ``types'' of \elf-frames are needed for the construction. He left open whether the problem remains \apx-hard when the input consists of only one type of \elf-frames or when the \elf-frames intersect a vertical line. We answer both questions affirmatively.
\begin{theorem}
\label{thm:apxEPG}
In the edge intersection model, the \mds~problem on \elf-frames of a single type is hard to approximate within a factor of 1.1377 even if all the \elf-frames intersect a vertical line from one side.
\end{theorem}

Furthermore, we show that even intersecting two lines does not help: the \mds~problem is \np-hard on \elf-frames in the edge intersection model even if every \elf-frame intersects a vertical line and a horizontal line. Observe that this is in contrast to the existence of the linear-time algorithm of Theorem~\ref{thm:xyCrossingAlgorithm} under the standard intersection model.

\subparagraph{Organization.} 
In Section \ref{sec:prelims}, we give some definitions and revisit some necessary background. We prove Theorem~\ref{thm:twoPlusEpsilon} and~\ref{thm:apxDiagonal} in Section~\ref{sec:diagonalPTAS}. The proofs of Theorem~\ref{thm:apxVertical} and~\ref{thm:xyCrossingAlgorithm} are given in Section~\ref{sec:yCrossing}. Finally, we show the results for the edge intersection model in Section~\ref{sec:edgeIntersection} and conclude the paper in Section~\ref{sec:conclusion}. Throughout this paper, the proofs of lemmas and theorems marked with $(*)$ are given in the full version of the paper due to space constraints.

\section{Preliminaries}\label{sec:prelims}
We denote the $x$- and $y$-coordinates of a point $p$ by $x(p)$ and $y(p)$, respectively. For two points $p$ and $q$, we denote the Euclidean distance between $p$ and $q$ by $\dis(p,q)$. Given a graph $G=(V,E)$, we denote the \elf-frame corresponding to a vertex $u\in V(G)$ by $\pat(u)$; we use $u$ and $\pat(u)$ interchangeably. We denote the corner of an \elf-frame $l$ by $\cor(l)$.

\subparagraph{Local search.} Consider an optimization problem in which the objective is to compute a feasible subset $S'$ of a ground set $S$ whose cardinality is minimum over all such feasible subsets of $S$. Moreover, it is assumed that computing some initial feasible solution and determining whether a subset $S'\subseteq S$ is a feasible solution can be done in polynomial time. The local search algorithm for a minimization problem is as follows. Fix some parameter $k$, and let $A$ be some initial feasible solution for the problem. In each iteration, if there are $A'\subseteq A$ and $M\subseteq S\setminus A$ such that $|A'|\leq k$, $|M|<|A'|$ and $(A\setminus A')\cup M$ is a feasible solution, then set $A:=(A\setminus A')\cup M$ and re-iterate. The algorithm returns $A$ and terminates when no such local improvement is possible.

Clearly, the local search algorithm runs in polynomial time. Let $\mathcal{B}$ and $\mathcal{R}$ be the solution returned by the algorithm and an optimal solution, respectively. We can assume that $\mathcal{B}\cap\mathcal{R}=\emptyset$; otherwise, we can remove the common elements of $\mathcal{B}$ and $\mathcal{R}$ and analyze the algorithm with the new sets, which guarantees that the approximation factor of the original instance is upper bounded by that of the new sets. The following result establishes the connection between local search technique and obtaining a $\mathsf{PTAS}$.
\begin{theorem}[\cite{ChanH12,MustafaR10}]
\label{thm:LSgivesPTAS}
Consider the solutions $\mathcal{B}$ and $\mathcal{R}$ for a minimization problem, and suppose that there exists a planar bipartite graph $H=(\mathcal{B}\cup\mathcal{R}, E)$ that satisfies the local exchange property, which is as follows: for any subset $\mathcal{B}'\subseteq\mathcal{B}$, $(\mathcal{B}\setminus\mathcal{B'})\cup N_H(\mathcal{B'})$ is a feasible solution, where $N_H(\mathcal{B'})$ denotes the set of neighbours of $\mathcal{B'}$ in $H$. Then, the local search algorithm yields a \ptas for the problem.
\end{theorem}

\section{Diagonal-intersecting Rectangles}
\label{sec:diagonalPTAS}
In this section, we prove Theorem~\ref{thm:twoPlusEpsilon} and~\ref{thm:apxDiagonal}. To prove Theorem~\ref{thm:twoPlusEpsilon}, we first give a \ptas for the problem when the rectangles are anchored at the diagonal from only one side and will then prove the theorem by applying the \ptas twice. 

Recall the class of intersection graphs of diagonal-anchored rectangles is same as that of diagonal-anchored \elf-frames~\cite{CatanzaroCFHHHS17}. As such, to simplify the presentation of the result, we prove Theorem~\ref{thm:twoPlusEpsilon} for \elf-frames.

\subsection{PTAS}
\label{subsec:ptas}
Suppose that we are given a set of \elf-frames each of which is anchored at the diagonal from above; let $G=(V, E)$ be the corresponding graph. Consider any two \elf-frames $\pat_1$ and $\pat_2$ that intersect each other. We say that $\pat_1$ and $\pat_2$ are \emph{coincident} if $x(\cor(\pat_1))$=$x(\cor(\pat_2))$. $\pat_1$ is said to intersect $\pat_2$ from left (resp. from below) if $x(\cor(\pat_1))\le x(\cor(\pat_2))$ (resp. $x(\cor(\pat_1))$ $\ge$ $x(\cor(\pat_2))$). Notice that $\pat_1$ intersects $\pat_2$ from left if and only if $\pat_2$ intersects $\pat_1$ from below. 

Consider the \mds~problem on $G$ and run the local search algorithm with $k:=c/\epsilon$ for some constant $c$. Let $\mathcal{B}$ be the solution returned by the local search algorithm and $\mathcal{R}$ denote an optimal solution. Consider the bipartite graph $H=(\mathcal{B}\cup \mathcal{R}, E')$ in which the edge set $E'$ is defined as follows. For any vertex $u\in V$, consider the set of all \elf-frames in $\{\mathcal{B}\cup\mathcal{R}\}$ that intersect $\pat(u)$ and let $(b_i, r_j)$, where $b_i\in\mathcal{B}$ and $r_j\in\mathcal{R}$, be a pair for which $\dis(\cor(b_i), \cor(r_j))$ (i.e., the Euclidean distance between their corners) is minimum over all such pairs. Then, $(b_i, r_j)\in E'$. We call $u$ a \emph{witness} vertex for the pair $(b_i, r_j)$. See Figure~\ref{fig:closestCorners} for an example of how an edge of $H$ is determined.

\begin{figure}[!h]
\centering
\includegraphics[width=0.50\textwidth]{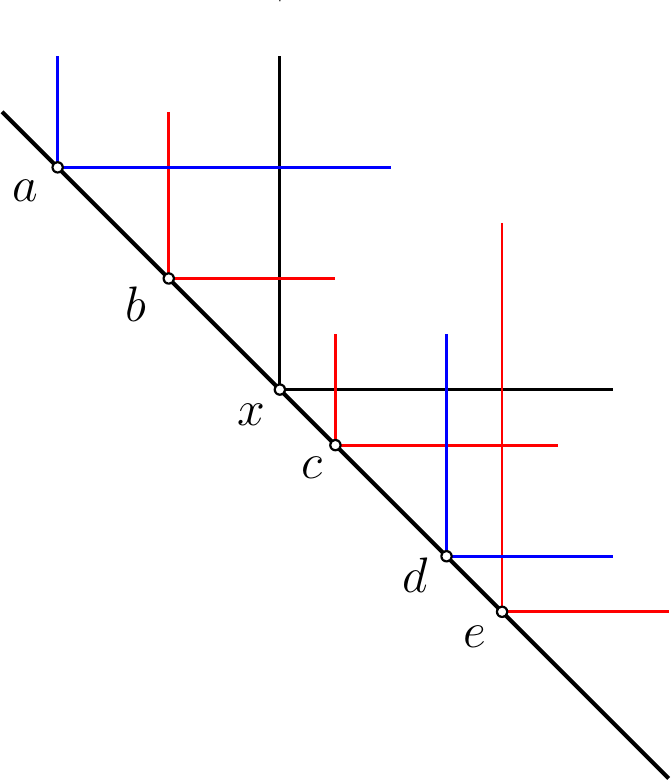}
\caption{The \elf-frames $a,d\in\mathcal{B}$ (blue) and $b,c,e\in\mathcal{R}$ (red) intersect the witness $x$. Since $\cor(d)$ and $\cor(e)$ are the closest pair among all possible red-blue corner pairs, we have $(d,e)\in E'$.}
\label{fig:closestCorners}
\end{figure}

In the following, we show that $H$ is planar (Lemma~\ref{lem:noTopEdgeCrossing}--\ref{lem:noTwoMixedCrossing}) and will then prove that $H$ satisfies the local exchange property (Lemma~\ref{lem:localExchange}). Then, by Theorem~\ref{thm:LSgivesPTAS}, our local search algorithm yields a \ptas for the problem on $G$. To distinguish between the edges of $G$ and those of $H$, we refer to the edges of $H$ as \emph{arcs}. Let $b_i\in \mathcal{B}$ and $r_j\in \mathcal{R}$ such that $(b_i,r_j)\in E'$. We say that $(b_i,r_j)$ is a \emph{top arc} if there is a witness $u$ for $(b_i, r_j)$, such that \emph{both} $\pat(b_i)$ and $\pat(r_j)$ intersect $\pat(u)$ from left; choose one such $u$ arbitrarily and denote $\pat(u)$ by $\wit(b_i, r_j)$. Otherwise, if there is a witness $v$ for $(b_i, r_j)$, such that \emph{both} $\pat(b_i)$ and $\pat(r_j)$ intersect $\pat(v)$ from below, $(b_i, r_j)$ is a \emph{down arc}; choose one such $v$ arbitrarily and denote $\pat(v)$ by $\wit(b_i, r_j)$. Otherwise, for any witness $w$ of $(b_i, r_j)$, $\pat(b_i)$ and $\pat(r_j)$ intersect $w$ from different sides; we call $(b_i,r_j)$ a \emph{mixed arc}, and choose one such $w$ arbitrarily and denote $\pat(w)$ by $\wit(b_i, r_j)$.

\begin{figure}[t]
\centering
\includegraphics[width=.80\textwidth]{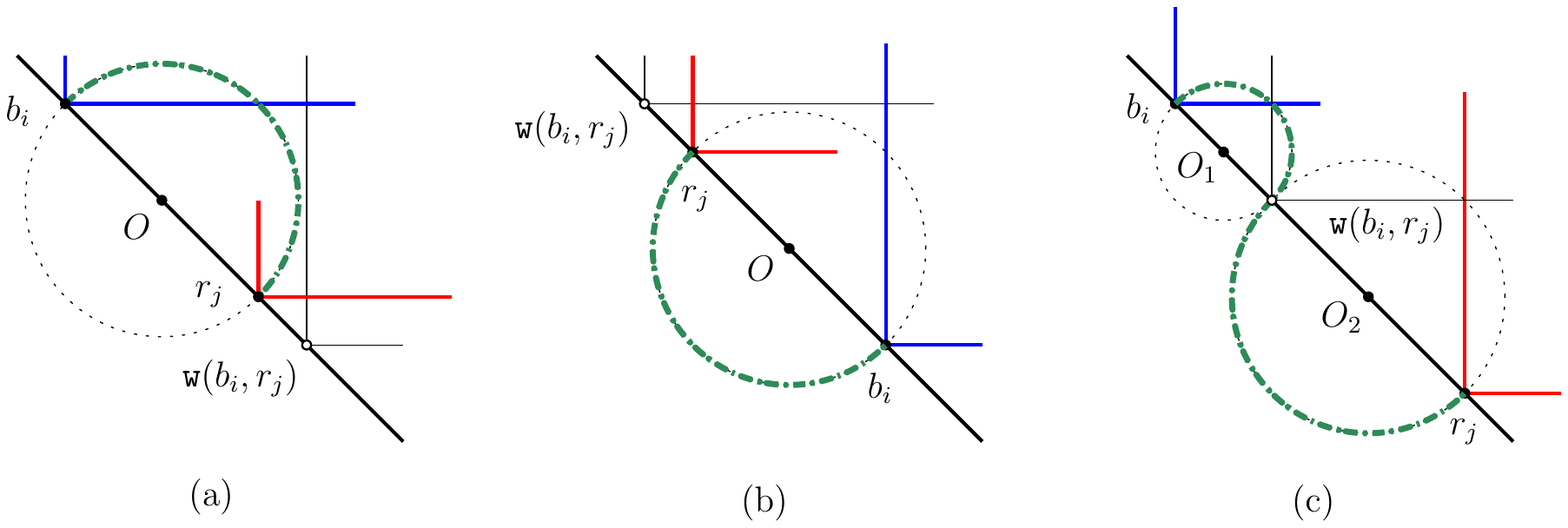}
\caption{Drawing of (a) a top arc, (b) a down arc, and (c) a mixed arc.}
\label{fig:edgeTypes}
\end{figure}

\subparagraph{Drawing of $H$.} To draw $H$, we map $u$ to $\cor(u)$ on the diagonal $D$ for all $u\in \mathcal{B}\cup \mathcal{R}$. To draw the arcs of $H$, the idea is to draw each arc either above or below the diagonal depending on whether it is a top arc or a down arc, respectively. A mixed arc is drawn in such a way that some part of the arc is drawn above and some part is drawn below the diagonal (and, hence crosses the diagonal). We next give the details of each case.

Let $b_i\in \mathcal{B}$ and $r_j\in\mathcal{R}$ such that $(b_i,r_j)\in E'$ and assume w.l.o.g. that $x(\cor(b_i))<x(\cor(r_j))$. Moreover, let
\[
O:=(\frac{x(\cor(b_i))+x(\cor(r_j))}{2}, \frac{y(\cor(b_i))+y(\cor(r_j))}{2}),
\]
and consider the circle $C$ centred at $O$ with radius $\dis(\cor(b_i),\cor(r_j))/2$. If the arc $(b_i,r_j)$ is a top arc (resp., down arc), then we draw it on the half circle of $C$ that lies above the diagonal $D$ (resp., that lies below the diagonal $D$) starting from $\cor(b_i)$ and ending at $\cor(r_j)$; see Figure~\ref{fig:edgeTypes}(a)---(b) for an illustration. A mixed arc is drawn in a slightly different way. For a mixed arc $(b_i, r_j)$, assume w.l.o.g. that $\pat(b_i)$ intersects $\wit(b_i,r_j)$ from left while $\pat(r_j)$ intersects $\wit(b_i,r_j)$ from below. Notice that $x(\cor(b_i))<x(\cor(\wit(b_i,r_j)))<x(\cor(r_j))$. To draw $(b_i,r_j)$, let
\[
O_1:=(\frac{x(\cor(b_i))+x(\cor(\wit(b_i,r_j)))}{2}, \frac{y(\cor(b_i))+y(\cor(\wit(b_i,r_j)))}{2}),
\]
\[
O_2:=(\frac{x(\cor(\wit(b_i,r_j)))+x(\cor(r_j))}{2}, \frac{y(\cor(\wit(b_i,r_j)))+y(\cor(r_j))}{2}),
\]
and consider the following two circles: the circle $C_1$ that is centred at $O_1$ and has the radius $\dis(\cor(b_i),\cor(\wit(b_i,r_j)))/2$, and the circle $C_2$ that is centred at $O_2$ with radius $\dis(\cor(\wit(b_i,r_j)),\cor(r_j))/2$. See Figure~\ref{fig:edgeTypes}(c). We draw the first part of arc $(b_i,r_j)$ on the half circle of $C_1$ that lies above the diagonal starting from $\cor(b_i)$ and ending at $\cor(\wit(b_i,r_j))$ and then the second part on the half circle of $C_2$ that lies below the diagonal starting from $\cor(\wit(b_i,r_j))$ and ending at $\cor(r_j)$. See Figure~\ref{fig:exampleGraphH} for an example.

\begin{figure}[!h]
\centering
\includegraphics[width=0.90\textwidth]{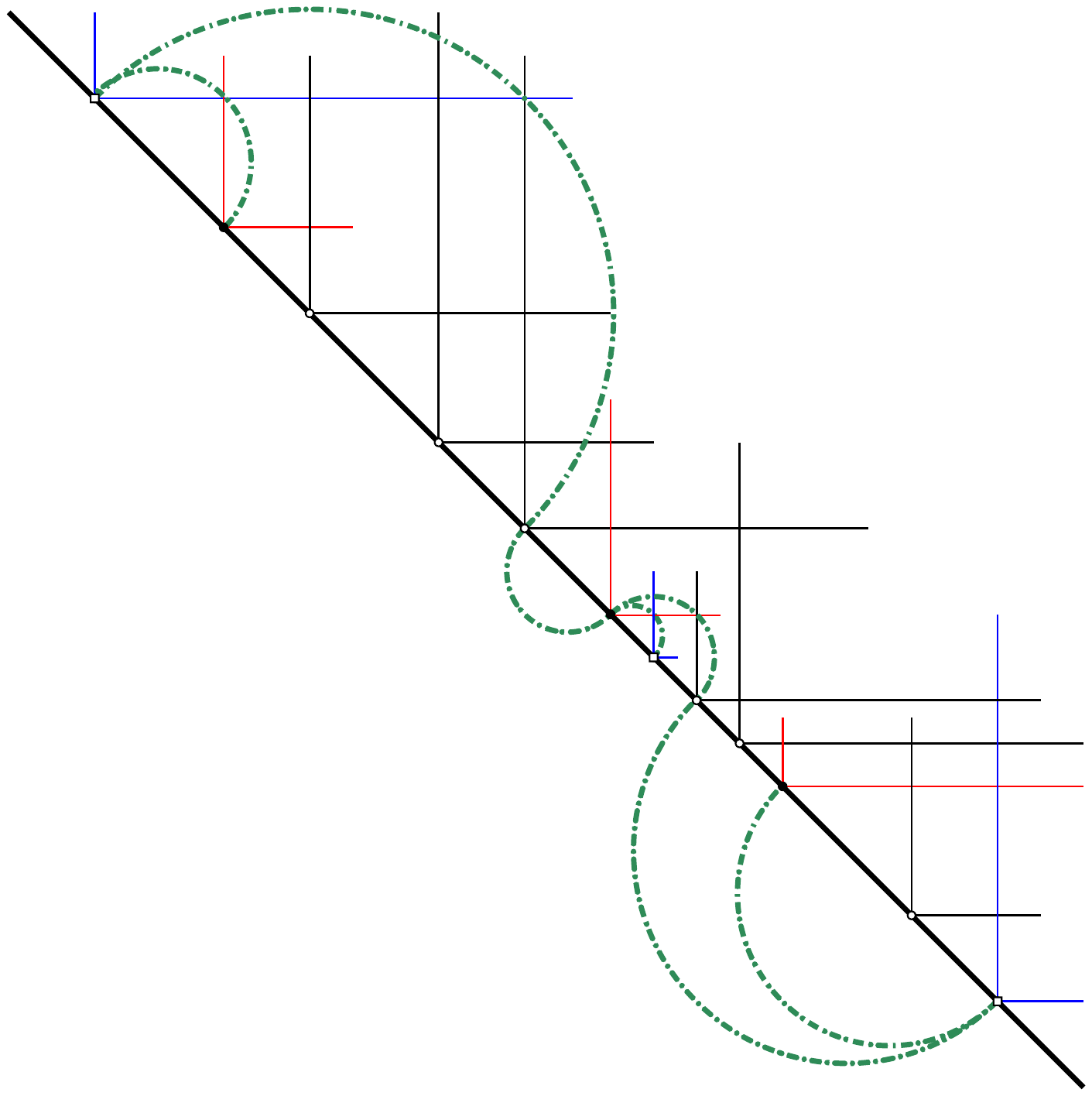}
\caption{A drawing of the graph $H$. The corner of the \elf-frames that are in $\mathcal{R}$ (resp., $\mathcal{B}$) are marked by black circles (resp., white squares).}
\label{fig:exampleGraphH}
\end{figure}

\subparagraph{Planarity of $H$.} Clearly, no top arc crosses a down arc (except perhaps at their endpoints). To show the planarity of $H$, we show that --- no top arc crosses another top arc, no down arc crosses another down arc, and no mixed arc crosses a top, a down or another mixed arc. 
\begin{lemma}
\label{lem:noTopEdgeCrossing}
No two top arcs in $H$ cross each other.
\end{lemma}
\begin{proof}
Suppose for a contradiction that two top arcs $(a,b), (c,d)\in E'$ cross each other, and w.l.o.g. assume that $x(a)<x(c)<x(b)<x(d)$. Since $(a,b)$ is a top arc, we must have $x(\cor(\wit(a,b)))\ge x(b)$; for a similar reason, we must have $x(\cor(\wit(c,d)))\ge x(d)$. We now consider two cases. (i) If $x(\cor(\wit(a,b)))\leq x(\cor(\wit(c,d)))$, then $\pat(c)$ must intersect $\wit(a,b)$, which is a contradiction because in that case we should have added $(a,c)$ or $(c,b)$ to $E'$ corresponding to $\wit(a,b)$ instead of $(a,b)$. (ii) If $x(\cor(\wit(a,b)))>x(\cor(\wit(c,d)))$, then $\pat(b)$ must intersect $\wit(c,d)$ --- this is also a contradiction because in that case we should have added $(c,b)$ or $(b,d)$ to $E'$ corresponding to $\wit(c,d)$ instead of $(c,d)$.
\end{proof}

\begin{figure}[t]
\centering
\includegraphics[width=0.85\textwidth]{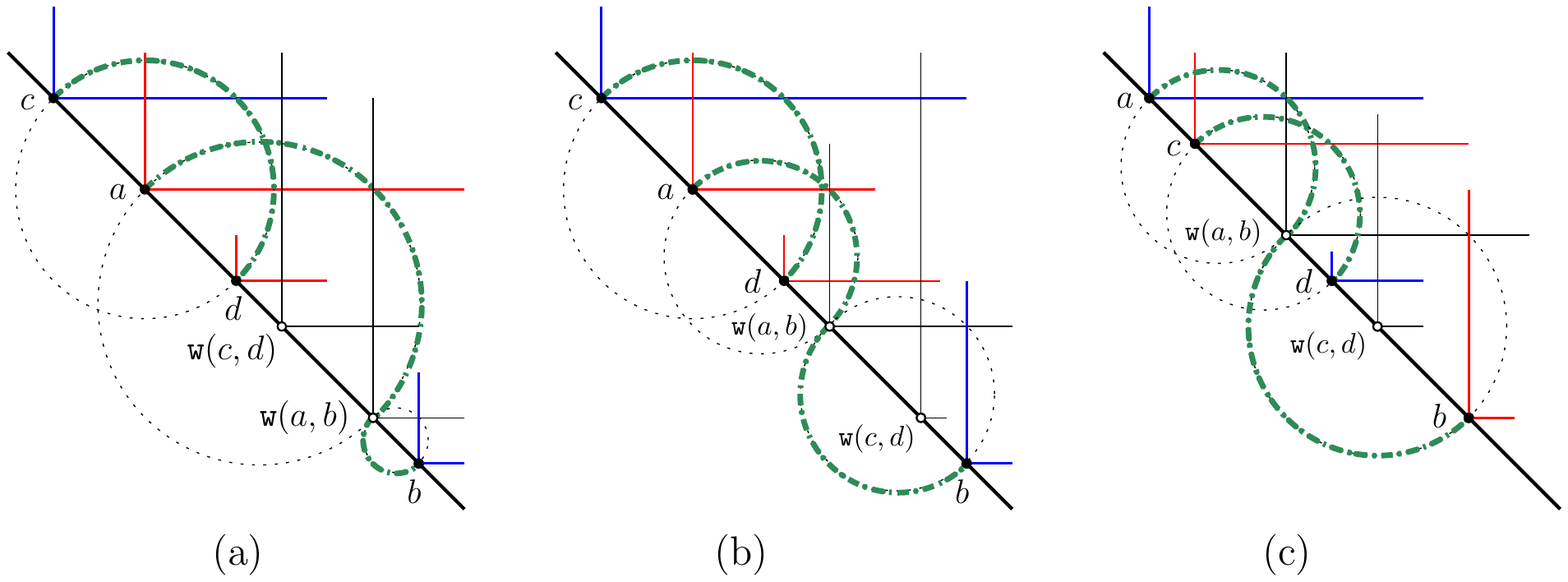}
\caption{An illustration in supporting the proof of Lemma~\ref{lem:noMixednoTopCrossing}.}
\label{fig:mixedTop}
\end{figure}

\begin{lemma}
\label{lem:noDownEdgeCrossing}
No two down arcs in $H$ cross each other.
\end{lemma}

\begin{proof}
 Suppose for a contradiction that two down arcs $(a,b), (c,d)\in E'$ cross each other, and w.l.o.g. assume that $x(a)<x(c)<x(b)<x(d)$. Since $(a,b)$ is a down arc, we must have $x(\wit(a,b))\le x(a)$; for a similar reason, we must have $x(\wit(c,d))\le x(c)$. We now consider two cases. (i) If $x(\wit(a,b))\geq x(\wit(c,d))$, then $\pat(c)$ must intersect $\wit(a,b)$, which is a contradiction because in that case we should have added either $(a,c)$ or $(c,b)$ to $E'$ corresponding to the witness $\wit(a,b)$ instead of the arc $(a,b)$. (ii) If $x(\wit(a,b))<x(\wit(c,d))$, then $\pat(b)$ must intersect $\wit(c,d)$ --- this is a contradiction because in that case we should have added either $(c,b)$ or $(b,d)$ to $E'$ corresponding to the witness $\wit(c,d)$, instead of the arc $(c,d)$.
\end{proof}

It remains to show that no mixed arc crosses a top arc, a down arc or another mixed arc.
\begin{lemma}
\label{lem:noMixednoTopCrossing}
No mixed arc crosses a top arc in $H$.
\end{lemma}
\begin{proof}
Suppose for a contradiction that a mixed arc $(a,b)$ crosses a top arc $(c,d)$ and assume w.l.o.g. that $x(a)<x(b)$ and $x(c)<x(d)$. If $x(c)<x(a)$, then we must have $x(a)<x(d)<x(\cor(\wit(a,b)))$ --- otherwise, the two arcs would not cross. First, we know that $x(\cor(\wit(c,d)))\ge x(d)$. If $x(\cor(\wit(c,d)))<x(\cor(\wit(a,b)))$, then $\pat(a)$ intersects $\wit(c,d)$ before intersecting $\wit(a,b)$, which is a contradiction because in that case we should have added either the top arc $(c,a)$ or the top arc $(a,d)$ to $E'$ corresponding to $\wit(c,d)$ instead of $(c,d)$; see Figure~\ref{fig:mixedTop}(a). If $x(\cor(\wit(c,d)))>x(\cor(\wit(a,b)))$, then $\pat(d)$ intersects $\wit(a,b)$ before intersecting $\wit(c,d)$, which is again a contradiction as we should have added either the top arc $(a,d)$ or the mixed arc $(d,b)$ to $E'$ corresponding to $\wit(a,b)$ instead of $(a,b)$. See Figure~\ref{fig:mixedTop}(b).

Now, suppose that $x(c)>x(a)$. Then, we must have $x(a)<x(c)<x(\cor(\wit(a,b)))<x(d)$ as the two arcs would not intersect otherwise. Since $(c,d)$ is a top arc, we must have $x(\cor(\wit(c,d)))\ge x(d)$ and so $\pat(c)$ would intersect $\wit(a,b)$ before intersecting $\wit(c,d)$; see Figure~\ref{fig:mixedTop}(c). This is a contradiction because we should have added either the top arc $(a,c)$ or the mixed arc $(c,b)$ to $E'$ corresponding to $\wit(a,b)$ instead of $(a,b)$.
\end{proof}

\begin{lemma}
\label{lem:noMixednoDownCrossing}
No mixed arc crosses a down arc in $H$.
\end{lemma}

\begin{proof}
 Suppose for a contradiction that a mixed arc $(a,b)$ crosses a down arc $(c,d)$ and assume w.l.o.g. that $x(\cor(a))<x(\cor(b))$ and $x(\cor(c))<x(\cor(d))$. First, suppose that $x(\cor(d))>x(\cor(b))$. Then, we must have $x(\cor(\wit(a,b)))<x(\cor(c))<x(\cor(b))$ --- otherwise, the two arcs would not cross. Since $(c,d)$ is a down arc, we know that $x(\cor(\wit(c,d)))\le x(\cor(c))$. If $x(\cor(\wit(c,d)))<x(\cor(\wit(a,b)))$, then $\pat(c)$ must intersect $\wit(a,b)$, which is a contradiction because we should have added either $(a,c)$ or $(c,b)$ to $E'$ corresponding to the witness $\wit(a,b)$, instead of the arc $(a,b)$; see Figure~\ref{fig:mixedDown}(a). If $x(\cor(\wit(c,d)))>x(\cor(\wit(a,b)))$, then $\pat(b)$ must intersect $\wit(c,d)$, which is a contradiction as in that case we should have added $(c,b)$ or $(b,d)$ to $E'$ corresponding to the witness $\wit(c,d)$, instead of the arc $(c,d)$; see Figure~\ref{fig:mixedDown}(b).

\begin{figure}[t]
\centering
\includegraphics[width=0.95\textwidth]{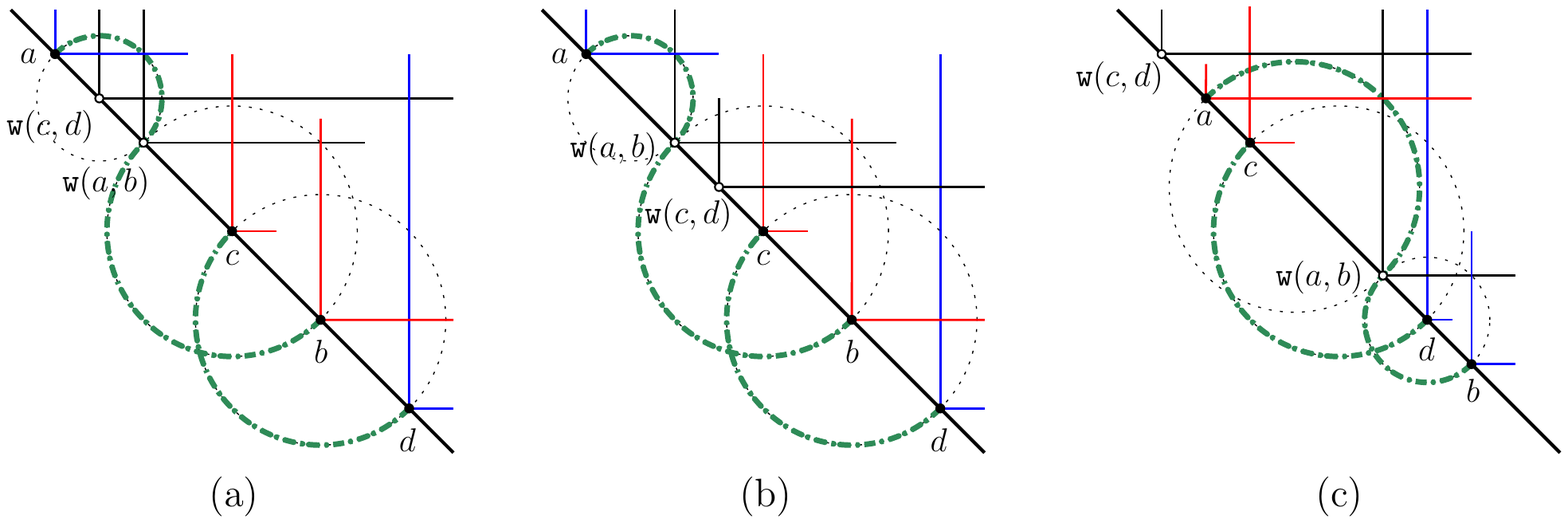}
\caption{An illustration in supporting the proof of Lemma~\ref{lem:noMixednoDownCrossing}.}
\label{fig:mixedDown}
\end{figure}

Now, suppose that $x(\cor(d))<x(\cor(b))$. Then, we must have that $x(\cor(c))<x(\cor(\wit(a,b)))<x(\cor(d))$ as otherwise the two arcs would not cross; see Figure~\ref{fig:mixedDown}(c). Since $(c,d)$ is a down arc, we must have $x(\cor(\wit(c,d)))\le x(\cor(c))$ and so $\pat(d)$ must intersect $\wit(a,b)$ --- this is a contradiction because in that case we should have added $(a,d)$ or $(d,b)$ to $E'$ corresponding to the witness $\wit(a,b)$, instead of the arc $(a,b)$.
\end{proof}

\begin{lemma}
\label{lem:noTwoMixedCrossing}
No two mixed arcs cross each other in $H$.
\end{lemma}

\begin{proof}
 Let $(a,b)$ and $(c,d)$ be two mixed arcs in $H$. First, $(a,b)$ and $(c,d)$ do not cross each other on the diagonal $D$ because crossing on the diagonal means that $\wit(a,b)$ and $\wit(c,d)$ are coincident and so we should have had only one arc among $\{a,b,c,d\}$ in $H$ for both $\wit(a,b)$ and $\wit(c,d)$. As such, we assume in the following that $\wit(a,b)$ and $\wit(c,d)$ are not coincident.

If $(a,b)$ and $(c,d)$ cross each other, then they can have at most two crossings, and in case of two, one is above and the other is below the diagonal. We show that no such types of crossings are possible. First, suppose that $(a,b)$ and $(c,d)$ cross each other above the diagonal, and assume w.l.o.g. that $x(a)<x(c)<x(\wit(a,b))<x(\wit(c,d))$. See Figure~\ref{fig:twoMixed}(a). Then, it is easy to see that $\pat(c)$ would intersect $\wit(a,b)$ before intersecting $\wit(c,d)$ and so we should have added either the top arc $(a,c)$ or the mixed arc $(c,b)$ to $E'$ corresponding to $\wit(a,b)$ instead of $(a,b)$ (see Figure~\ref{fig:twoMixed}(a)). Now, suppose that $(a,b)$ and $(c,d)$ cross each other below the diagonal and assume w.l.o.g. that $x(\wit(a,b))<x(\wit(c,d))<x(b)<x(d)$. See Figure~\ref{fig:twoMixed}(b). Then $\pat(b)$ must intersect $\wit(c,d)$ before intersecting $\wit(a,b)$ and so we should have added either the down arc $(b,d)$ or the mixed arc $(c,b)$ to $E'$ corresponding to the witness $\wit(c,d)$ instead of the arc $(c,d)$ (see Figure~\ref{fig:twoMixed}(b)). This concludes the proof.

\begin{figure}[t]
\centering
\includegraphics[width=0.60\textwidth]{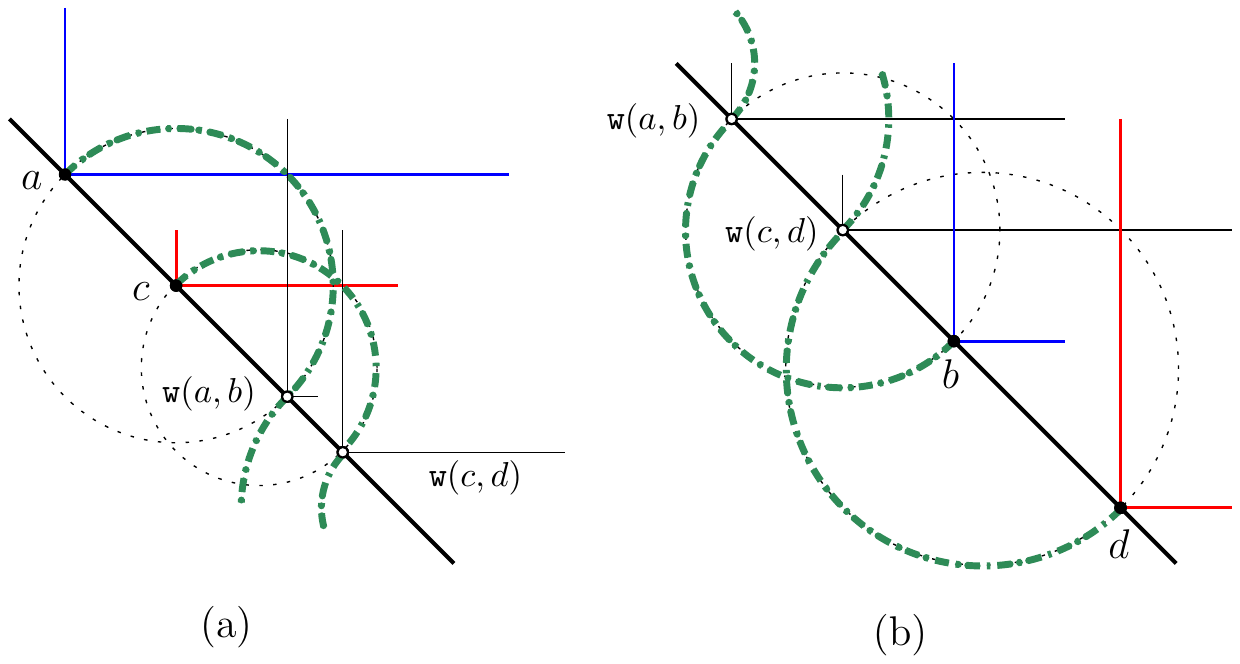}
\caption{An illustration in support of the proof of Lemma~\ref{lem:noTwoMixedCrossing}.}
\label{fig:twoMixed}
\end{figure}
\end{proof}

By Lemma~\ref{lem:noTopEdgeCrossing}---\ref{lem:noTwoMixedCrossing}, it follows that $H$ is a planar graph. The following lemma along with the planarity of $H$ gives a \ptas for the problem on $G$.
\begin{lemma}
\label{lem:localExchange}
Graph $H=(\mathcal{B}\cup\mathcal{R}, E)$ satisfies the local exchange property.
\end{lemma}
\begin{proof}
To prove the lemma, it is sufficient to show that for any vertex $u\in V$, there are $b_i\in\mathcal{B}$ and $r_j\in\mathcal{R}$ such that both $\pat(b_i)$ and $\pat(r_j)$ intersect $\pat(u)$ and $(b_i,r_j)\in E'$. Take any vertex $u\in V$ and let $S\subseteq \mathcal{B}\cup\mathcal{R}$ be the set of all \elf-frames that intersect $\pat(u)$. Notice that $S\cap\mathcal{B}\neq\emptyset$ and $S\cap\mathcal{R}\neq\emptyset$ because each of $\mathcal{B}$ and $\mathcal{R}$ is a feasible solution to the \mds~problem on $G$. Now, consider $b\in S\cap\mathcal{B}$ and $r\in S\cap\mathcal{R}$ for which $\dis(\cor(b),\cor(r))$ is minimum over all such pairs $(b,r)$. We know by definition of $H$ that $(b, r)\in E'$, which proves the lemma.
\end{proof}

\subsection{Proof of Theorem~\ref{thm:twoPlusEpsilon}}
\label{subsec:twoSided}
We are now ready to prove Theorem~\ref{thm:twoPlusEpsilon}. Here, the rectangles are anchored at the diagonal from both sides; let $G=(V,E)$ denote the corresponding graph.
\begin{proof}[Proof of Theorem~\ref{thm:twoPlusEpsilon}]
Let $\{X, Y\}$ be a partition of the vertices of $G$ such that the \elf-frames corresponding to vertices in $X$ (resp. $Y$) are anchored at $D$ from above (resp. below). By abusing notation, we refer to these two sets of \elf-frames also as $X$ and $Y$. For any $\epsilon>0$, set $\epsilon':=\epsilon/2$. We apply the \ptas of Section~\ref{subsec:ptas} with parameter $c/\epsilon'$ to the \elf-frames in $X$ and $Y$ independently, and let $S_X$ and $S_Y$ be the solutions returned by the algorithm, respectively. We return $S:=S_X\cup S_Y$ as the final solution. Let $\opt$, $\optX$ and $\optY$ denote an optimal solution for the \mds~problem on the \elf-frames in $X\cup Y$ (i.e., graph $G$), $X$ and $Y$, respectively.

Consider the solution $\opt$. Let $S\subseteq \opt$ be a minimum size set of \elf-frames that dominates all the \elf-frames in $X$. If there is an \elf-frame $P\in S$ that is in $Y$, then $P$ can only dominate those \elf-frames in $X$ that are anchored at the diagonal at the same point as $P$, in which case we can replace $P$ by one of those \elf-frames from $X$. As such, there exists a set $S'\subseteq X$ of size at most $|S|\le |\opt|$ that dominates the \elf-frames of $X$. It follows that $|\optX|\leq |\opt|$. Similarly, one can show that $|\optY|\leq |\opt|$. Now, by using the result from Section~\ref{subsec:ptas}, we have $|S_X|\leq (1+\epsilon')|\optX|$ and $|S_Y|\leq (1+\epsilon')|\optY|$. Then,
\[
|S|\leq |S_X|+|S_Y|\leq (1+\epsilon')|\optX|+(1+\epsilon')|\optY|\leq 2(1+\epsilon')|\opt|=(2+\epsilon)|\opt|,
\]
which completes the proof of the theorem.
\end{proof}

\subsection{Proof of Theorem~\ref{thm:apxDiagonal}}
\label{sec:apxDiagonal}
Recall that Theorem~\ref{thm:apxDiagonal} claims \apx-hardness of the problem on \elf-frames in the diagonal-intersecting case. We show a gap-preserving reduction from the \mds~problem on \emph{circle graphs}, which is known to be \apx-hard~\cite{DamianP06}. Recall that a graph is called a circle graph, if it is the intersection graph of chords of a circle. Take any circle graph $G=(V, E)$ with $n$ vertices, and consider a geometric representation of $G$. By a closer look at the \apx-hardness proof~\cite{DamianP06}, we can assume that no two chords share an endpoint; that is, there are exactly $2n$ distinct points on the circle determining the endpoints of the chords. Cut the circle at an arbitrary point $p$ and consider the ordering $M:=\langle p_1,p_2,\dots,p_{2n}\rangle$ of the endpoints of chords visited in counter-clockwise along the circle starting at $p$. Now, let $D$ be the diagonal line $y=(2n+1)-x$. Consider each endpoint $p_i$ (where $1\leq i\leq 2n$) in the order given by $M$. Then, we map each point $p_i$ to the point $((2n+1)-i, i)$ on $D$. Let $e:=(p_j, p_k)$ be a chord of $G$. Then, the \elf-frame corresponding to $e$ is the unique \elf-frame that lies below $D$ and connects the two points mapped on $D$ corresponding to $p_j$ and $p_k$; see Figure~\ref{fig:apxHardness} for an example. Let $G'=(V',E')$ be the graph corresponding to the constructed \elf-frames. Clearly, $D$ has slope -1 and the construction can be done in polynomial time. Moreover, since the $2n$ endpoints of chords of the input representation are pairwise distinct, the points mapped to the line $D$ are in general position; that is, no two such points have the same $x$- or the same $y$-coordinates. The following lemma completes the proof of Theorem~\ref{thm:apxDiagonal}.

\begin{figure}
\centering
\includegraphics[width=0.7\textwidth]{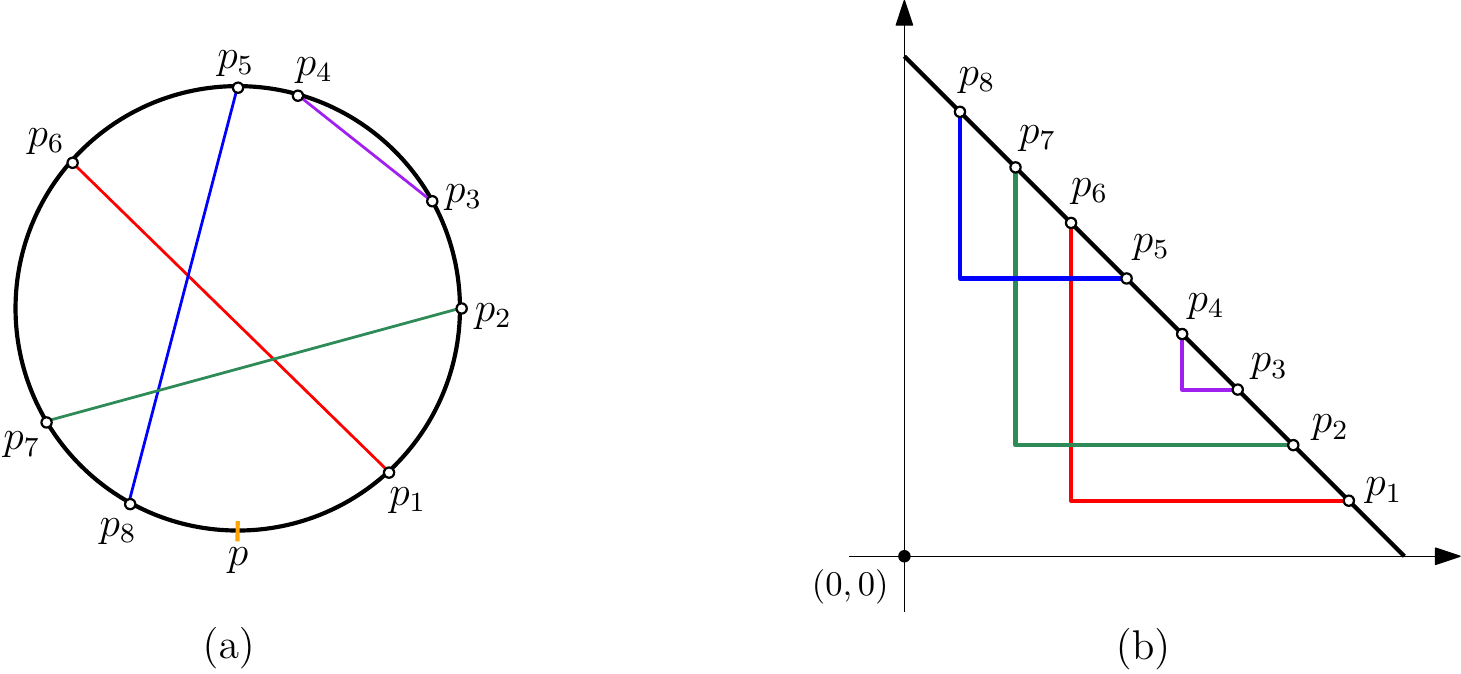}
\caption{An illustration supporting the proof of Theorem~\ref{thm:apxDiagonal}.}
\label{fig:apxHardness}
\end{figure}

\begin{lemma}
\label{lem:apxHardnessDiagonal}
$G$ has a dominating set of size $k$ if and only if $G'$ has a dominating set of size $k$.
\end{lemma}
\begin{proof}
First, we show that there is a one-to-one correspondence between the vertices of $G$ and those of $G'$ such that $(i,j)\in E$ if and only if $(i,j)\in E'$. The one-to-one correspondence between the vertices of these two graphs is clear from the construction. Now, suppose that $(i, j)\in E$. This means that the endpoints of vertices $i$ and $j$ appear in $M$ as either $\langle p_i,p_j,p_{i'},p_{j'}\rangle$ or $\langle p_j,p_i,p_{j'},p_{i'}\rangle$, where $p_i$ and $p_{i'}$ with $i<i'$ (resp., $p_j$ and $p_{j'}$ with $j<j'$) correspond to the endpoints of $i$ (resp., $j$). By the mapping used in the construction, this ordering is preserved on $D$ (when going from $(2n+1,0)$ to $(0,2n+1)$ along $D$) and so the \elf-frames $\pat(i)$ and $\pat(j)$ intersect each other below $D$; that is, $(i,j)\in E'$. Conversely, if $(i,j)\notin E$, then their endpoints appear in $M$ as either $\langle p_i,p_j,p_{j'},p_{i'}\rangle$ or $\langle p_j,p_i,p_{i'},p_{j'}\rangle$, where $p_i$ and $p_{i'}$ with $i<i'$ (resp., $p_j$ and $p_{j'}$ with $j<j'$) correspond to the endpoints of $i$ (resp., $j$). Consequently, this ordering is preserved on $D$ by the mapping (when going from $(2n+1,0)$ to $(0,2n+1)$ along $D$) and so $\pat(i)$ and $\pat(j)$ do not intersect each other inside $D$; that is, $(i,j)\notin E'$. It is now straightforward to see that $G$ has a dominating set of size $k$ if and only if $G'$ has a dominating set of size $k$, which is clearly a gap-preserving reduction.
\end{proof}
\subparagraph*{Remark.} Using a similar reduction in the other direction one can show that, for every intersection graph $G$ of a set of diagonal-intersecting \elf-frames that intersect each other only below the diagonal, one can find a circle graph $G'$ such that $G$ has a dominating set of size $k$ if and only if $G'$ has a dominating set of size $k$. As one can obtain a $2+\epsilon$-approximation for \mds~on \emph{circle graphs} \cite{Damian-IordacheP02}, \mds~can be approximated within a factor of $2+\epsilon$ in this special case of diagonal-intersecting \elf-frames. Note that this is the version which we have just proved to be \apx-hard.

\section{Vertical-intersecting L-frames}
\label{sec:yCrossing}
In this section, we prove Theorem~\ref{thm:apxVertical} and~\ref{thm:xyCrossingAlgorithm}. To prove Theorem~\ref{thm:apxVertical}, we first show that the \mds~problem is \apx-hard even when each \elf-frame intersects a vertical line. The proof is essentially the same as that of Theorem~\ref{thm:apxDiagonal}, but 
here we replace the diagonal $D$ with a quarter of a circle and then extend the horizontal segment of each \elf-frame until it hits a vertical line that is placed far to the right. 
\begin{figure}[t]
\centering
\includegraphics[width=0.9\textwidth]{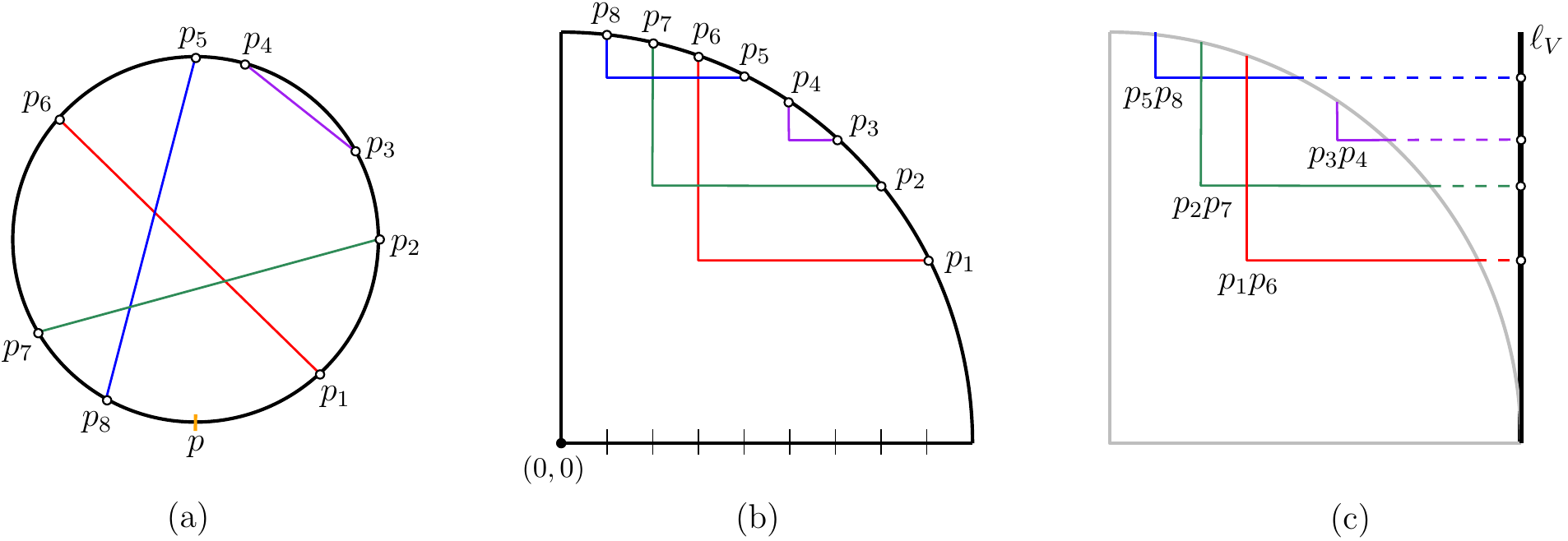}
\caption{An illustration supporting the proof of Theorem~\ref{thm:apxVertical}.}
\label{fig:circleHardness}
\end{figure}
Next we show the gap-preserving reduction from the \mds~problem on \emph{circle graphs}, which is known to be \apx-hard~\cite{DamianP06}. Recall that a graph is called a circle graph, if it is the intersection graph of chords of a circle. Take any circle graph $G=(V, E)$ with $n$ vertices, and consider a geometric representation of $G$. Recall that we can assume no two chords share an endpoint; that is, there are exactly $2n$ distinct points on the circle determining the endpoints of the chords. Cut the circle at an arbitrary point $p$ and consider the ordering $M:=\langle p_1,p_2,\dots,p_{2n}\rangle$ of the endpoints of chords visited in counter-clockwise along the circle starting at $p$. Now, let $C$ be a circle with radius one that is centred at the origin $(0, 0)$ of the Cartesian coordinate system. Consider each endpoint $p_i$ (where $1\leq i\leq 2n$) in the order given by $M$, and let $x_i:=i/(2n+1)$. Then, we map $p_i$ to the point $(x_i, \sqrt{1-x_i^2})$ on the top right quadrant of circle $C$. Let $e:=(p_j, p_k)$ be a chord of $G$. Then, the \elf-frame corresponding to $e$ is the unique \elf-frame that lies inside $C$ and connects the two points mapped on $C$ corresponding to $p_j$ and $p_k$; see Figure~\ref{fig:circleHardness}(b) for an illustration. Finally, we complete the construction by extending the horizontal segment of every \elf-frame to the right until it intersects the vertical line $x=1$. See Figure~\ref{fig:circleHardness}(c). Clearly, the construction can be done in polynomial time, and since the $2n$ endpoints of chords of the input representation are pairwise distinct, the points mapped to the circle $C$ are in general position; that is, no two such points have the same $x$- or the same $y$-coordinates. Now, one can easily prove a lemma similar to Lemma~\ref{lem:apxHardnessDiagonal} for this case as well, and so the \apx-hardness of the problem follows.

\begin{figure}[t]
\centering
\begin{subfigure}{.31\textwidth}
  \centering
  \includegraphics[width=.9\linewidth]{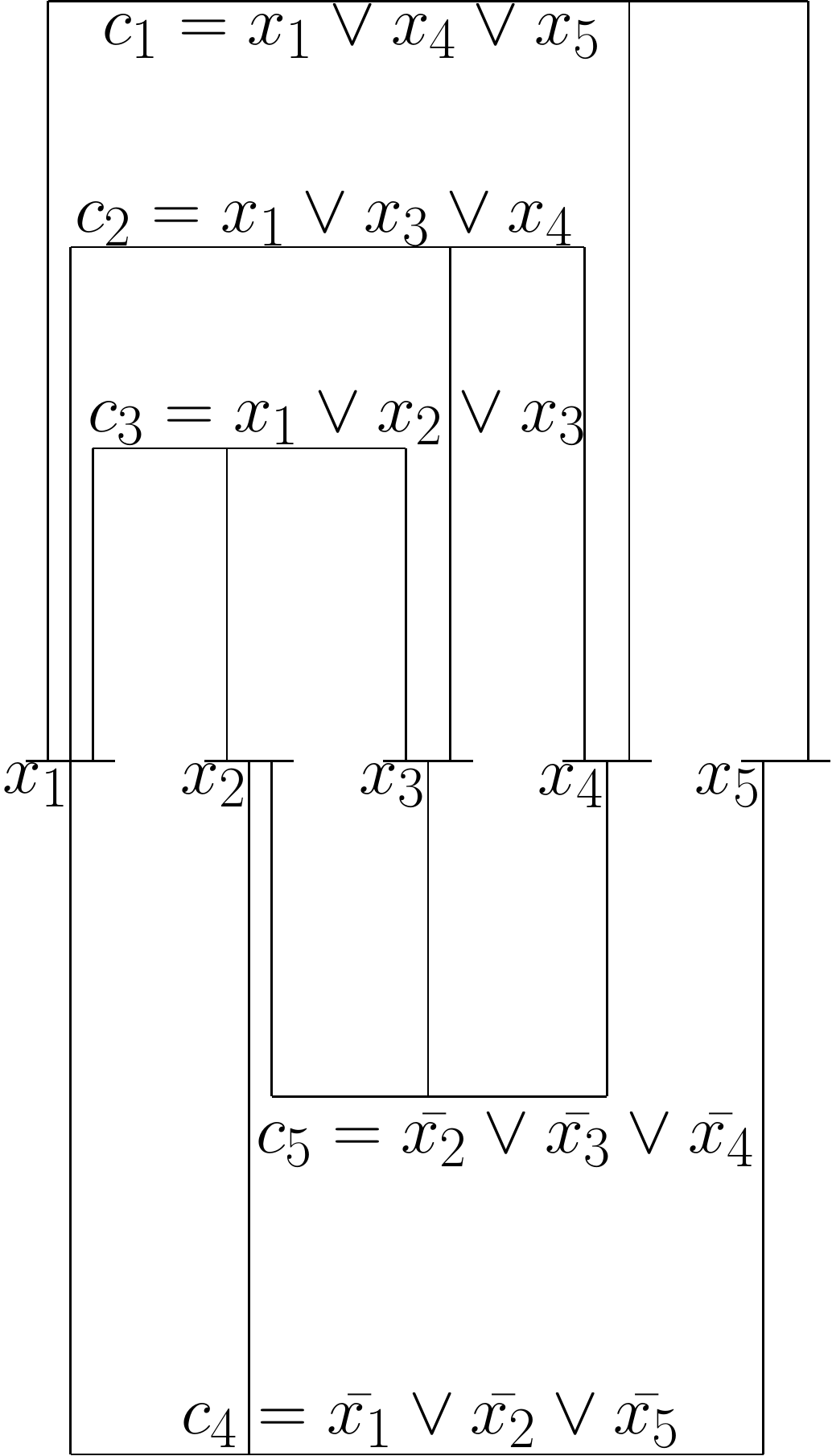}
  \label{fig:sub1}
\end{subfigure}
\hspace{5mm}
\begin{subfigure}{.55\textwidth}
  \centering
  \includegraphics[width=.95\linewidth]{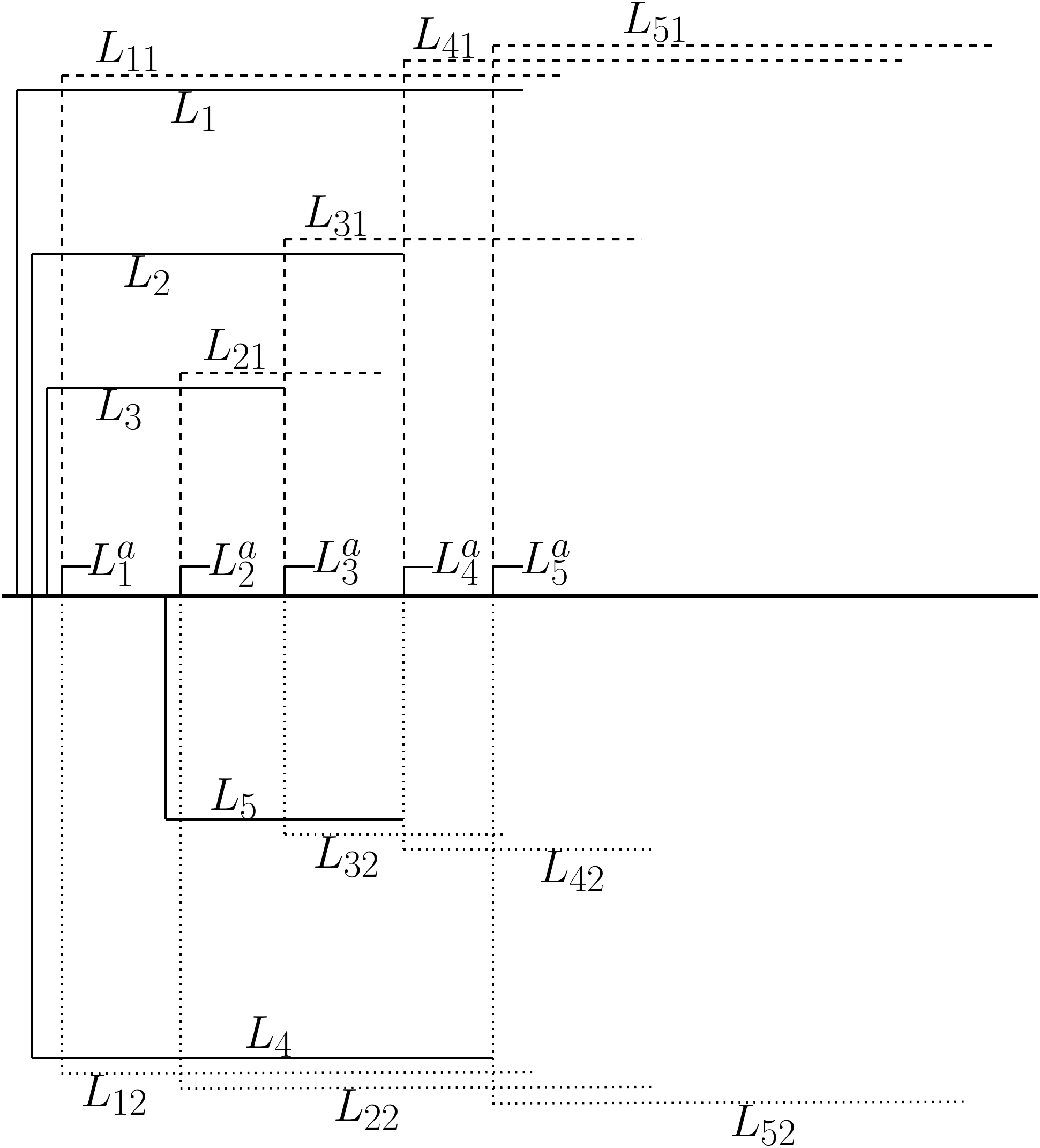}
  \label{fig:sub2}
\end{subfigure}
\caption{A drawing of an instance of Planar Monotone Rectilinear 3SAT (left) and the corresponding instance of the \mds~problem before rotation (right). The variable \elf-frames above (resp. below) the $x$-axis are shown in dashed (resp. dotted) style. The clause \elf-frames and the auxilliary \elf-frames are shown in normal style.}
\label{fig:nphardverttouch}
\end{figure}

To complete the proof of Theorem~\ref{thm:apxVertical}, we show that the \mds~problem is \np-hard on \elf-frames even if the horizontal and vertical segments of every \elf-frame have the same length. The reduction is from a variant of the 3SAT problem. For any 3SAT instance, one can define a bipartite graph on the clauses and the variables which we refer to as the incidence graph. For any clause $c_j$, if $c_j$ contains a literal corresponding to a variable $v_i$, the edge $(v_i,c_j)$ is added to the graph. A drawing of a planar incidence graph is called \textit{planar rectilinear} if it has the following properties. Each variable vertex is drawn as a horizontal segment on the $x$-axis and the clause vertices are drawn as horizontal segments above and below the $x$-axis. For each edge $(v_i,c_j)$, the horizontal segment corresponding to $c_j$ has a vertical connection to the horizontal segment corresponding to $v_i$. Moreover, this vertical connection does not intersect any other horizontal segments. See the left figure of Figure \ref{fig:nphardverttouch} for an example. An instance of the 3SAT problem is called \emph{monotone} if, for every clause in the instance, the literals of the clause are either all positive (called a \emph{positive} clause) or all negative (called a \emph{negative} clause). A planar rectilinear drawing of an incidence graph is called \emph{monotone} if it corresponds to a monotone 3SAT instance such that all the positive (resp., negative) clauses are drawn above (resp., below) the $x$-axis. In the Planar Monotone Rectilinear 3SAT problem the incidence graph of any instance has a planar rectilinear and monotone drawing. The problem is known to be \np-hard even when the drawing is given~\cite{BergK12}. Given an instance of the Planar Monotone Rectilinear 3SAT, we construct a set of \elf-frames such that the horizontal and vertical segments of each \elf-frame have the same length and they all intersect a vertical line.

Let $n$ be the number of variables. Due to the hierarchical structure of the clauses lying above and below the $x$-axis, respectively in the given drawing, one can modify the drawing in a way such that for each clause, the length of the horizontal segment corresponding to it is the same as the length of the vertical segments corresponding to it (see the left figure of Figure \ref{fig:nphardverttouch}). Let $T$ be the new drawing. We construct the instance $I$ of \mds~from $T$ in the following way. For each clause $c_j$, removing all the vertical connections except the leftmost one creates an \elf-frame. We add this \elf-frame (denoted by $L_j$) to our instance. For each variable $x_i$, we add two \elf-frames: one (denoted by $L_{i1}$) above and the other (denoted by $L_{i2}$) below the $x$-axis. The vertical segments of these two \elf-frames are the maximum length vertical connections corresponding to $x_i$ lying above and below $x$-axis, respectively. For the variable \elf-frames, their horizontal segments lie on the right of their vertical segment. Moreover, for any variable \elf-frame above (resp., below) the $x$-axis, the corner is the topmost (resp., bottommost) point of the vertical segment. The two \elf-frames corresponding to each variable are shifted accordingly so that they intersect at a point on the $x$-axis (see Figure \ref{fig:nphardverttouch}). Also, for each variable $x_i$, we add an auxiliary \elf-frame $L_i^a$ that intersects only the two \elf-frames corresponding to it as shown in Figure~\ref{fig:nphardverttouch}. Note that in the constructed instance, one endpoint of each \elf-frame lies on the $x$-axis. Now it could be the case that, for some $i, j$ and $k\in \{1,2\}$, an $L_{ik}$ intersects $L_j$ though $c_j$ does not contain any literal of $x_i$. To get rid of these intersections, we extend each $L_{ik}$ vertically and then horizontally (to maintain the symmetry) so that it does not intersect any $L_j$ such that $c_j$ does not contain $x_i$. Note that, as we do not modify the $L_j$'s, if $c_j$ contains the literal $x_i$ (resp. $\bar{x_i}$), $L_{i1}$ (resp. $L_{i2}$) still intersects $L_j$. Finally, we rotate the created instance clockwise by 90$^{\circ}$ with respect to the origin. Hence all the \elf-frames in the constructed instance now intersect the vertical line $x=0$ or the $y$-axis, and thus together they form an instance of the problem with vertical-intersecting \elf-frames. Lemma~\ref{lem:nphardsymmetric} below completes the proof of the second part of Theorem~\ref{thm:apxVertical}. We first need the following observation whose proof is immediate from the construction.
\begin{obs}
\label{obs:nphardverttouch}
For all $i,j$, $L_{i1}$ (resp., $L_{i2}$) intersects $L_j$ if and only if the clause $c_j$ contains the literal $x_i$ (resp., $\bar{x_i}$). 
\end{obs}

\begin{lemma}
\label{lem:nphardsymmetric}
The instance $I$ has a dominating set of size $n$ if and only if the 3SAT instance is satisfiable.
\end{lemma}
\begin{proof}
Suppose $I$ has a dominating set $D$ of size $n$. We construct a satisfiable assignment in the following way. Consider any variable $x_i$. As $L_i^a$ intersects only $L_{i1}$ and $L_{i2}$, one of these three \elf-frames must be in the dominating set. The size of $D$ being $n$, only one of the above mentioned three \elf-frames can be in $D$ and no clause \elf-frame can be in $D$. If $L_{i1}$ is in $D$, we set $x_i$ to be true. Otherwise, we set $x_i$ to be false. Now consider any clause $c_j$. Let $x_i$ be a variable such that a \elf-frame corresponding to $x_i$ is in $D$ that intersects $L_j$. Note that there exists such an $x_i$. Now if $L_{i1}$ is the chosen \elf-frame that dominates $L_j$, then $L_{i1}$ is in $D$ and by Observation~\ref{obs:nphardverttouch} $c_j$ must be a positive clause. It follows that, we have set $x_i$ to be true and thus $c_j$ is satisfied. Similarly, one can show that if $L_{i2}$ is the chosen \elf-frame that dominates $L_j$, then also $c_j$ is satisfied.

Now, suppose that we are given a satisfiable assignment of the 3SAT instance. We construct a dominating set in the following way. For any variable $x_i$, if $x_i$ is set to be true, then we select $L_{i1}$. If $x_i$ is set to be false, then we select $L_{i2}$. We claim that the selected \elf-frames form a dominating set. It is easy to see that the selected \elf-frames dominate the variable and auxiliary \elf-frames. Now, consider any clause \elf-frame $L_j$. There must be a variable $x_i$ that satisfies $c_j$. If $c_j$ is positive, the literal $x_i$ must be set to true and we select $L_{i1}$. By Observation~\ref{obs:nphardverttouch}, $L_{i1}$ intersects $L_j$, and thus $L_j$ is being dominated. If $c_j$ is negative, one can similarly show that $L_j$ is also being dominated.
\end{proof}

We now prove Theorem~\ref{thm:xyCrossingAlgorithm}. To this end, we show that the intersection graph of \elf-frames that inersect a vertical and a horizontal line is a permutation graph and so \mds~is linear-time solvable on such a graph~\cite{ChaoHL00}. Geometrically, a graph $G$ is a permutation graph if there are two embeddings of its vertices on two parallel lines such that, when connecting a vertex from the first line to the same vertex on the other line using a line segment, the edge set of the graph is realized exactly by the intersections of those segments. That is, two vertices appear in different order on the parallel lines (and, so their line segments intersect) whenever they are adjacent in the graph. By taking the two orderings in which the endpoints of the \elf-frames in the input graph intersect the two lines, we can construct the geometric representation of a permutation graph having the same edges. Hence, Theorem~\ref{thm:xyCrossingAlgorithm} follows.

\section{Edge Intersection Model}
\label{sec:edgeIntersection}
In this section, we show our results for the edge intersection model~\cite{GolumbicLS09}.

\subparagraph{Proof of Theorem~\ref{thm:apxEPG}.} We show a reduction from the Vertex Cover problem. In Vertex Cover, we are given a graph $G=(V,E)$ with $n$ vertices and the goal is to find a subset $V'\subseteq V$ such that for any edge $(v_i,v_j)\in E$, $V'\cap\{v_i,v_j\}\ne \emptyset$. Let $VC(G)$ be the size of any minimum size vertex cover of $G$. As noted in \cite{AwasthiCKS15}, from the work of Dinur and Safra \cite{dinursafra} the following theorem can be derived.
\begin{theorem}[Dinur and Safra~\cite{dinursafra}]
\label{thm:VC}
Let $1/3 <p < p_{max}=(3-\sqrt{5})/2$ and $q=4p^3-3p^4$. For any constant $\epsilon > 0$, given any unweighted graph $G=(V,E)$ with bounded degrees, it is \np-hard to distinguish between ``Yes'': VC(G) $< (1-p+\epsilon)|V|$, and ``No'': VC(G) $> (1-q-\epsilon)|V|$.
\end{theorem}

Given the vertex cover graph $G=(V,E)$, we construct an instance $I$ of \mds. In the instance $I$, all the \elf-frames lie on the left of the line $x=0$ and intersect the line, and thus each has the $\llcorner$ shape. The construction is as follows. For each vertex $v_i$, we have an \elf-frame $L_i$ that intersects $x=0$ at the point $(0,2i)$ as shown in Figure \ref{fig:apxepg}. Note that these \elf-frames do not intersect each other. For each edge $(v_i,v_j)$ with $i < j$, we have an \elf-frame $L_{ij}$ that intersects $x=0$ at $(0,2j)$, and thus intersects (i.e., has a common horizontal grid edge) with $v_j$. It also intersects (has a common vertical grid edge) with $v_i$ (see Figure \ref{fig:apxepg}). Also, for each vertex $v_i$, we have two auxiliary \elf-frames $A_{i1}$ and $A_{i2}$ that intersect $x=0$ at $(0,2n+i)$, and $A_{i1}$ shares a vertical grid edge with $L_i$. Note that the only \elf-frame $A_{i2}$ intersects is $A_{i1}$. Also, all the \elf-frames intersect the vertical line $x=0$. We first need the following observations.

\begin{figure}[t]
\centering
\begin{subfigure}{.25\textwidth}
  \centering
  \includegraphics[width=.75\linewidth]{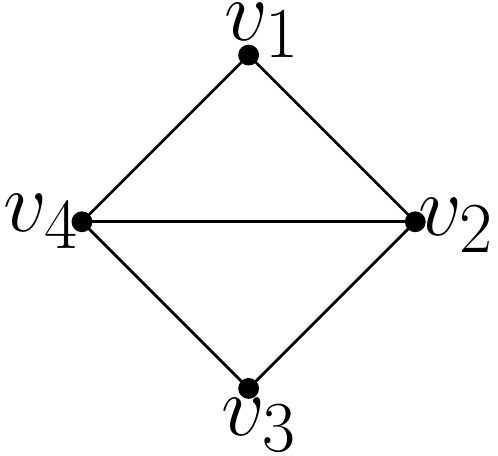}
\end{subfigure}
\hspace{15mm}
\begin{subfigure}{.35\textwidth}
  \centering
  \includegraphics[width=.60\linewidth]{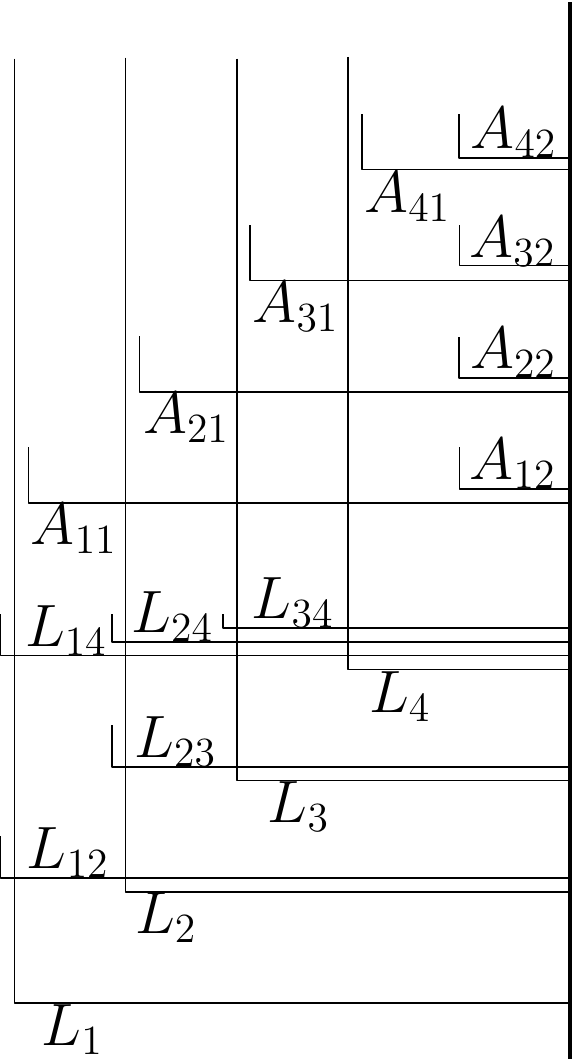}
\end{subfigure}
\caption{An example graph with 4 vertices (left) and its corresponding instance of the \mds~problem (right).}
\label{fig:apxepg}
\end{figure}

\begin{obs}
\label{obs:edge}
For $1\leq i<j\leq n$, the only \elf-frames that $L_{ij}$ intersects are $L_i, L_j$ and $L_{tj}$, where $(v_t,v_j) \in E$ with $t< j$. 
\end{obs}

\begin{obs}
\label{obs:vertex}
The only \elf-frames that get dominated by the \elf-frames in $\{A_{i1}:1\le i\le n\}$ are $L_i$ and $A_{i2}$ for $1\leq i\le n$.
\end{obs}

\begin{obs}
\label{obs:auxiliary}
Any dominating set contains at least one of $A_{i1}$ and $A_{i2}$ for each $1\leq i\leq n$. Also for any dominating set of size $k$, there is a dominating set $D$ of size at most $k$ such that $D$ contains $A_{i1}$ for all $1\leq i\leq n$ and does not contain any $A_{i2}$ for $1\leq i\leq n$.
\end{obs}

\begin{lemma}
\label{lem:reduction}
$G$ has a vertex cover of size $k$ if and only if $I$ has a dominating set of size $k+n$. 
\end{lemma}
\begin{proof}
Suppose $G$ has a vertex cover of size $k$. We construct a dominating set $D$ of size $k+n$ in the following way. At first add all the \elf-frames $\{A_{i1}:1\le i\le n\}$ to $D$. Then for each $v_i$ in the vertex cover, add $L_i$ to $D$. Clearly $|D|=k+n$. By Observation \ref{obs:vertex}, $D$ dominates any $L_i$ and $A_{i2}$ for $1\le i \le n$. Now consider any $L_{ij}$. Note that we have added at least one of $L_i$ and $L_j$ to $D$, and hence $D$ dominates $L_{ij}$ as well.

Now, suppose we have a dominating set $D'$ of size $k+n$. By Observation \ref{obs:auxiliary}, we can assume that $D'$ contains $A_{i1}$ for all $i$ and does not contain any $A_{i2}$. Thus $D'$ contains $k$ \elf-frames corresponding to the vertices and the edges. Now suppose $D'$ contains an edge \elf-frame $L_{ij}$. By Observation \ref{obs:edge} and \ref{obs:vertex}, the only \elf-frames that cannot get dominated by $D'$ by the removal of $L_{ij}$ from $D'$ are $L_{tj}$, where $(v_t,v_j) \in E$ with $t< j$. Thus if we replace $L_{ij}$ by $L_j$ in $D'$, $D'$ still remains a dominating set. Thus, we can assume w.l.o.g. that $D'$ does not contain any $L_{ij}$. Now we construct a subset $V'\subseteq V$ by selecting vertices corresponding to the vertex \elf-frames contained in $D'$. Clearly $|V'|=k$. We claim that $V'$ is a vertex cover. Consider any edge $(v_i,v_j)$. Then by Observation \ref{obs:vertex}, at least one of $L_i$ and $L_j$ must be in $D'$. It follows that at least one of $v_i$ and $v_j$ is present in $V'$ which finishes the proof. 
\end{proof}

By Theorem~\ref{thm:VC} and Lemma~\ref{lem:reduction}, we have the following lemma. 
\begin{lemma}
Let $1/3 <p < p_{max}=(3-\sqrt{5})/2$ and $q=4p^3-3p^4$. For any constant $\epsilon > 0$, given an input graph (in the edge intersection model) $H$ with $O(n)$ vertices, it is \np-hard to distinguish between ``Yes'': $H$ has a dominating set of size $< (2-p+\epsilon)n$, and ``No'': The size of any dominating set of $H$ is $> (2-q-\epsilon)n$.
\end{lemma}

Therefore, the gap between the sizes of the minimum dominating set in the ``yes'' and the ``no'' cases approaches $(2-4p^3+3p^4)/(2-p)\approx 1.1377$ for $p=p_{max}$. Hence, Theorem~\ref{thm:apxEPG} follows.

\begin{figure}[t]
\centering
\includegraphics[width=0.55\textwidth]{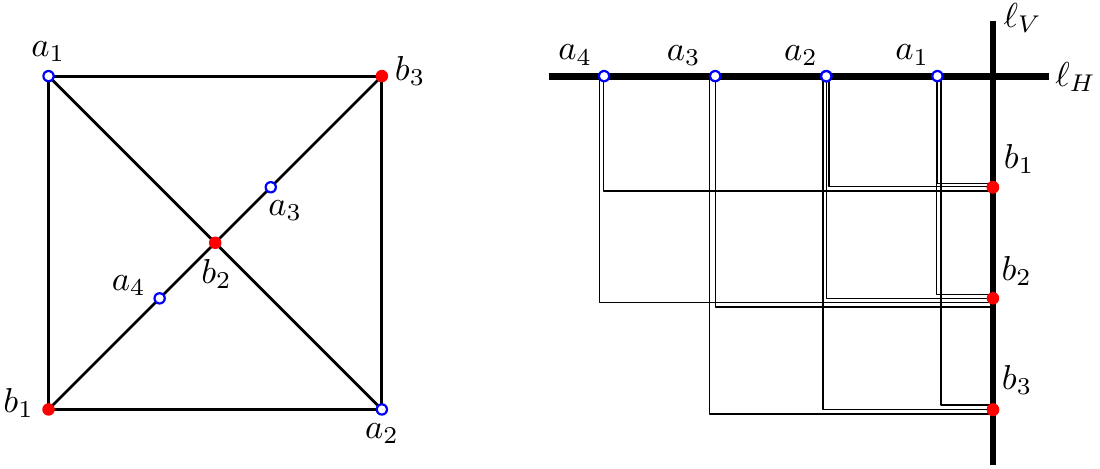}
\caption{A bipartite graph on 7 vertices (left) and its corresponding representation in the edge intersection model (right).}
\label{fig:hardness}
\end{figure}

\subparagraph{\np-hardness.} We show a reduction from the \emph{Edge Dominating Set} problem that is known to be \np-hard, even on planar bipartite graphs~\cite{HortonK93}. Recall that the objective of the Edge Dominating Set problem on a graph is to choose a minimum-cardinality set $S$ of edges of the graph such that every edge not in $S$ shares at least one endpoint with one edge in $S$. Given a planar bipartite graph $G=(A\cup B, E)$, we construct an intersection graph $G'$ of \elf-frames in polynomial time such that $G$ has an edge dominating set of size $k$ if and only if $G'$ has a dominating set of size $k$ (in the edge intersection model). To this end, let $A=\{a_1,\dots,a_r\}$ and $B=\{b_1,\dots,b_s\}$, and for each vertex $a_i\in A$ (resp., vertex $b_j\in B$), consider the point $(-i, 0)$ (resp., point $(0,-j)$) of the Cartesian coordinate system in the plane. Then, for each edge $(a_i,b_j)\in E(G)$, add to $V(G')$ the unique \elf-frame that connects $(-i,0)$ to $(0,-j)$ and has the point $(-i,-j)$ as its corner. See Figure~\ref{fig:hardness}. Let $G'$ be the resulting graph. Clearly, every \elf-frame in $G'$ intersects a vertical and a horizontal line and it can be constructed in polynomial time. Moreover, notice the one-to-one correspondence between the edges of $G$ and the \elf-frames in $G'$. It is now easy to see that $G$ has an edge dominating set of size $k$ if and only if $G'$ has a dominating set of size $k$ (in the edge intersection model).

\section{Conclusion}
\label{sec:conclusion}
In this paper, we considered the \mds~problem on the intersection graphs of rectangles and \elf-frames. Among several other approximation and hardness results, we gave a $(2+\epsilon)$-approximation algorithm for the problem on diagonal-anchored rectangles, which was based on a \ptas for when the rectangles are anchored from one side. However, the complexity of the latter problem remains open. The problem is \np-hard when the rectangles are anchored from both sides; is the problem \apx-hard? Or, does the problem admit a $\mathsf{PTAS}$? Even, finding an algorithm with approximation factor better than $(2+\epsilon)$ remains open.

\bibliographystyle{plain}
\bibliography{ref}

\end{document}